\documentclass[referee]{sn-jnl}
\usepackage{amsmath, amsthm, amscd, amsfonts, amssymb, graphicx, color, subcaption, longtable,pict2e}
\usepackage{caption}
\usepackage{mathtools}
\usepackage{hyperref}
\usepackage{MnSymbol}
\usepackage{float}
\usepackage{doi}
\usepackage{listings}
\usepackage{appendix}

\jyear{2021}%

\theoremstyle{thmstyleone}%
\newtheorem{theorem}{Theorem}[section]
\newtheorem{proposition}[theorem]{Proposition}%

\theoremstyle{thmstyletwo}%

\theoremstyle{thmstylethree}%
\newtheorem{definition}{Definition}[section]%

\numberwithin{equation}{section}
\newtheorem{solution}{Solution}[section]%
\renewenvironment{proof}{{\bfseries Proof}}{\qed}

\begin{document}

\title[An Ergodic Model of Monetary Velocity]{Log-Ergodic Dynamics in Stochastic Monetary Velocity: Theoretical Insights and Economic Implications}


\author*[1]{\fnm{Kiarash} \sur{Firouzi}}\email{kiarashfirouzi91@gmail.com}

\author[2]{ \fnm{Mohammad} \spfx{Jelodari} \sur{Mamaghani}}\email{j\_mamaghani@atu.ac.ir}


\affil*[1]{\orgdiv{Department of Mathematics}, \orgname{Allameh Tabataba'i University}, \orgaddress{\street{Dehkadeh Olympic}, \city{Tehran}, \postcode{1489684511}, \state{Tehran}, \country{Iran}}}

\affil[2]{\orgdiv{Department of Mathematics}, \orgname{Allameh Tabataba'i University}, \orgaddress{\street{Dehkadeh Olympic}, \city{Tehran}, \postcode{1489684511}, \state{Tehran}, \country{Iran}}}



\abstract{
We suggest employing log-ergodic processes to simulate the velocity of money in an ergodic manner. Our approach sheds light on economic behavior, policy implications, and financial dynamics by maintaining long-term stability. By bridging theory and practice, the partially ergodic model helps analysts and policymakers comprehend and forecast velocity of money. The empirical analysis, using historical U.S. GDP and money supply data, demonstrates the model’s effectiveness in capturing the long-term stability of the velocity of money. Key findings indicate that the log-ergodic model offers superior predictive power compared to traditional models, making it a valuable tool for policymakers to control economic factors in vital situations.
	

}




\maketitle
\section{Introduction}\label{sec1}
The velocity of money, a key economic measure, estimates the frequency with which currency units are exchanged for goods and services in an economy. Traditionally, it is measured as the ratio of nominal GDP to money supply. However, this static depiction obscures the intrinsic dynamic and stochasticity of monetary transactions \cite{ger}.

The velocity of money is a vital factor in economic stability. Traditional models frequently fail to represent its dynamic nature, especially in the context of long-term predictability and stability. Despite substantial study, there is still a gap in accurately simulating these elements. Inflation management and economic stability are top priorities for politicians and economists globally. The velocity of money is critical in understanding and regulating various economic variables. Adjusting this ratio can have an impact on inflation rates, economic growth, and overall financial stability. To address this, we propose a unique technique that represents velocity of money as a somewhat stochastic process, reflecting the essence of its fluctuating behavior across time. This study offers the log-ergodic process, a unique technique that fills a gap by offering a more robust framework for understanding the long-term evolution of monetary mobility. This work uses log-ergodic processes to bridge the gap between theoretical models and actual economic applications, providing fresh insights into the stability and predictability of monetary mobility. This mathematical approach enables us to examine the long-term stability and predictability of velocity of money, even in the face of short-term economic changes and uncertainty. Using ergodic theories, we show that the time average of the logarithm of economic variables like transaction volumes and money supply converges to their expected value, assuring long-term stability. The motivation behind this research is to provide a deeper understanding of how velocity of money behaves under stochastic conditions and how to effectively adjust the ratio of GDP to money supply to achieve desired economic outcomes which is essential for accurate economic forecasting and policy-making.

Our findings add to the current literature by giving a more sophisticated view of velocity of money via the lenses of stochastic calculus and ergodic theory. It provides an effective instrument for economists and financial analysts to evaluate and anticipate the effects of monetary policy changes and economic cycles on the velocity of money. Furthermore, our approach bridges the gap between theoretical constructions and actual economic applications, allowing policymakers to make better informed decisions.

Ergodic theory, based on statistical mechanics and dynamical systems theory, examines probabilistic system evolution over time \cite{1}. Ergodic theory is fundamentally concerned with the idea of ergodicity, which refers to the property of a system in which its time-average behavior converges to its ensemble-average behavior over large time scales. This fact has major consequences for understanding the behavior of financial markets. Investors strive to make judgments based on historical evidence and probabilistic projections \cite{15}.

The complex and dynamic structure of economic systems necessitates the use of an ergodic method to analyzing the velocity of money. Traditional models sometimes assume a static or deterministic view of economic indicators, which can result in oversimplifications and errors in interpreting and forecasting economic activity. Here is why the ergodic method is important:

Economic systems are naturally dynamic, undergoing fluctuations and changes throughout time \cite{chi}. An ergodic method allows us to describe the velocity of money as an evolving process, incorporating real-world unpredictability and mobility throughout the economy. Ergodic theory enables stable and predictable long-term behavior of economic variables, even during short-term volatility (\cite{1}). This is critical for making sound policy decisions aimed at achieving long-term economic goals. Traditional models frequently fail to consider the probabilistic character of economic interactions \cite{mast}. The ergodic method includes stochastic processes, which provide a more accurate portrayal of the unpredictability and uncertainty inherent in economic activity. Economists can foresee economic trends more accurately if they grasp the ergodic features of the velocity of money.  This is especially valuable for forecasting the impact of monetary policy changes and developing future economic strategies. Policymakers use accurate economic information to influence monetary policy \cite{orp}. The ergodic approach of estimating the velocity of money provides a solid foundation for assessing the effectiveness of policy interventions and directing the development of new policies.

The rest of the paper is organized as follows: \\
Section \ref{sec2} provides a concise overview of the key developments and scholarly contributions to the field, setting the stage for the research presented in the paper. In section \ref{sec3}, we present a detailed explanation of the concepts and methodologies that form the basis of our approach to modeling the velocity of money. Section \ref{sec4} provides a detailed explanation of the steps taken to construct and validate a stochastic model for the velocity of money using log-ergodic processes. In section \ref{sec5}, we demonstrate the practical application and validation of the log-ergodic model of monetary velocity. Section \ref{sec6} provides an overview for reflecting on the broader implications of the research, acknowledging its limitations, and suggesting avenues for future inquiry.
Finally, in section \ref{sec7}, we conclude the paper and suggest some directions for future research.
 
\section{Literature Review}\label{sec2}
For centuries, economists have studied the velocity of money, which has evolved through various schools of thought. This section examines the literature on the velocity of money, including classical models and the rise of stochastic and ergodic techniques in economic dynamics. The stochastic and dynamic behaviors of the velocity of money are frequently ignored in the existing literature, which results in models that are less useful in real-world scenarios. By presenting a log-ergodic process that represents the intrinsic stochastic nature and long-term stability of the velocity of money, this work fills the aforementioned gap. This method gives policymakers useful tools in addition to improving theoretical understanding.

Neo-classical economists such as Irving Fisher established the equation of exchange, which included the velocity of money as a fundamental component. Fisher's findings paved the way for later studies of money circulation within an economy. Classical models frequently assumed a constant velocity, which failed to account for the changes found in real-world economies \cite{tav}.

With the emergence of contemporary finance theory, economists began to appreciate the significance of randomness and uncertainty in economic activity. Stochastic models have become popular tools for studying financial markets and economic indicators \cite{2}. Pioneers such as Louis Bachelier and subsequently Paul Samuelson used stochastic processes to simulate price fluctuations and market dynamics, paving the way for stochastic modeling of monetary velocity \cite{sam,jar}. Paper \cite{17} investigates how the velocity of money affects inflation dynamics by estimating the Phillips curve derived from a New Keynesian model in which money is introduced via transactions technology. Luca Benati explores the relationship between the velocity of money (specifically M1) and the natural interest rates \cite{18}. Paper \cite{19} states that velocity of money significantly drops following negative output and expected inflation shocks, and exhibits regime-switching behavior. The paper \cite{20} by Serletis implies that the velocity of money exhibits chaotic behavior and can be modeled using random walks and breaking trend functions, highlighting its complex and unpredictable nature.

Ergodic theory, which originated in statistical physics, was applied to economics to investigate systems that change over time. Mathematicians like Andrey Kolmogorov and George Birkhoff created the mathematical underpinnings of ergodic theory, which were later applied to economic models by economists like John von Neumann. The theory established a framework for comprehending the long-term average behavior of dynamic systems, particularly economic systems \cite{1}.

Mean-ergodic processes are critical for understanding the velocity of money and its link to economic issues. Ergodic theory in stochastic processes relates long-run time averages to phase averages, offering insights into developing stationary processes \cite{leo}. Ergodicity is crucial in economic models, and recent findings provide necessary and sufficient conditions for economic dynamics \cite{1}. Ergodic theory studies the long-term behavior of systems, such as the velocity of money, and emphasizes the preservation of certain energy forms \cite{kami}. Paper \cite{16} seeks to clarify this specific relation between the idea of ergodicity which is drawn from statistical mechanics and its role in and for economics and finance. Paper \cite{21} implies that chaotic and ergodic behavior can emerge in deterministic economic models, suggesting that the inherent unpredictability in economic variables may be due to intrinsic model dynamics rather than external randomness. Incorporating monetary stock-flow data into input-output linkages and characterizing the flow of money as a Markov chain allows for an analysis of the velocity of money circulation within an economic equilibrium framework, independent of labor-theoretic or marginalist assumptions \cite{ser}. Dynamical systems theory examines the features of money and velocity series, offering insights into monetary aggregates and their dynamics over time \cite{cou}.

The use of log-ergodic models in economics is rather new. These models enable the examination of economic variables such as velocity of money in a multiplicative form, which is more natural for many economic processes. The log transformation stabilizes variance and allows for the application of ergodic theorems, offering insights into the long-term behavior of economic indicators \cite{fir}.

We integrated stochastic calculus and ergodic theory to develop advanced models of economic dynamics \cite{fir}. We've started using these complex mathematical methods to simulate the velocity of money, realizing the need for models that can capture its intrinsic fluctuation and deliver more accurate forecasts.

\section{Theoretical Framework}\label{sec3}

The theoretical framework for analyzing the velocity of money using stochastic and log-ergodic processes is based on the junction of economic theory, probability, and statistical mechanics. This section describes the fundamental notions and mathematical constructions that support our method for modeling monetary velocity.

\subsection{Stochastic Processes and Economic Indicators}
Our approach is based on the idea that economic indicators, such as the velocity of money, are best described as stochastic processes. These mechanisms explain the inherent unpredictability and uncertainty in economic interactions. The stochastic nature of these variables is shown by a probability distribution or a stochastic differential equation that characterizes their evolution over time \cite{2}.

Ergodicity describes a system's time average converging to its ensemble average \cite{1}. In economic modeling, we concentrate on log-ergodicity, which applies the notion to the logarithms of economic variables. This method is useful for positive stochastic processes like geometric Brownian motion and economic indicator growth rates. It is proven that the positive stochastic processes made by jump diffusion processes are partially ergodic. Partial ergodicity means that the descriptive model of the market shows some degree of mean-ergodicity hidden from the market participants \cite{fir}. 

Stochastic GDP growth is defined as the unpredictable changes in a country's Gross Domestic Product (GDP) over time. These fluctuations can be attributed to a number of causes, including changes in consumer behavior, investment levels, government policies, and external economic situations.
Economists use stochastic models to forecast GDP growth by incorporating random variables and probabilities. The Stochastic Growth Model is a popular model that contains elements such as the production function and stochastic processes \cite{22}.

As an example, consider an economy where the production function is given by: 
$$Y(t)=F\big(K(t),L(t),z(t)\big),$$
where $Y(t)$ is the GDP at time $t$, $K(t)$ is the capital stock, $L(t)$ is the labor input, and $z(t)$ is the stochastic productivity term. If $z(t)$ is a stochastic process, such as a Markov chain, it brings uncertainty into GDP growth. This means that future GDP growth is uncertain and is determined by the stochastic process.

\subsection{Mathematical Formulation of Log-Ergodic Processes}
A log-ergodic process occurs when the logarithm of a stochastic process, $Y(t) = \log(X(t))$, shows mean-ergodic behavior \cite{fir}. Mathematically, this indicates that for practically every realization of the process, the time average of $Y(t)$ converges to the expected value as time approaches infinity. i.e.

$$ \lim_{T \to \infty} \frac{1}{T} \int_0^T Y(t) \, dt = \mathbb{E}[Y(t)].$$

This property assures that the process's long-term behavior remains predictable and stable in the face of short-term variations.

From \cite{fir} we have the following definitions:
\begin{definition}(EMO)\label{deferc}
	Let $W_t$ be a standard Wiener process and $\beta>\frac{3}{2}$. For all $t,s\in[0,T]$, we define the ergodic maker operator of the process $Y_t^\prime$ as
	\begin{equation*}
		\xi_{t-s,W_{t-s}}^{\beta}[Y_t^\prime]:=0\cdot Y_0^\prime+\frac{W_T}{T^{\beta}}\cdot D_{t-s}+\frac{1}{T^{\beta}}\cdot R_{t-s},
	\end{equation*}
	From now on, for all $t>s$, we denote the length of the time interval $[s,t]$ by $\delta=t-s$. Therefore, we have 
	\begin{equation}\label{xi}
		\xi_{\delta,W_\delta}^{\beta}[Y_t^\prime]:=0\cdot Y_0^\prime+\frac{W_T}{T^{\beta}}\cdot D_\delta+\frac{1}{T^{\beta}}\cdot R_\delta,
	\end{equation}
where $\beta$ is called the inhibition degree parameter.
\end{definition}
\begin{definition}(Log-ergodic process)\label{logergodic}
	The positive stochastic process $X_t$ is log-ergodic, if its log process, $Y_t=\ln(X_t)$, satisfies
	\begin{equation}\label{limer}
		\overline{<Y>}:=\lim_{T\rightarrow \infty}\frac{1}{T}\int_0^T (1-\frac{\tau}{T})\mathbf{Cov}_{yy}(\tau)d\tau=0, \quad \forall\tau\in[0,T].
	\end{equation}
	Where $\mathbf{Cov}_{yy}(\tau)$ is the covariance of  $Y_\tau$.
\end{definition}
\begin{definition}(Partial ergodicity)
	The positive stochastic process $X_t$ is partially ergodic if 
	$Z_\delta=\xi_{\delta,W_\delta}^\beta[Y_t]$ satisfies \ref{limer}.
\end{definition}

Some stochastic processes are not log-ergodic. In such instances, we use the ergodic maker operator (EMO) on the process to observe the mean-ergodic behavior hidden in the original dataset and we refer to such a process as partially ergodic \cite{fir}.

\subsection{Application to Monetary Velocity}
When extending this concept to the velocity of money, we describe it as a stochastic process that represents the pace at which money flows in the economy. The velocity of money, $V(t)$, has historically been defined as the ratio of nominal GDP to money supply. However, in our model, we regard both GDP and the money supply as stochastic processes controlled by economic forces. We define $V(t)$ as:

\begin{equation}\label{V}
V(t) = \frac{X(t)}{M(t)},
\end{equation}

where $X(t)$ is the nominal GDP and $M(t)$ is the money supply at time $t$.

To verify that the velocity of money is partially ergodic (see \cite{fir}), we apply the log transformation and the EMO  to $V(t)$ and evaluate the resultant process, \(Z_\delta^v=\xi_{\delta,W_\delta}^\beta[\log(V(t))] \), where $\xi_{\delta,W_\delta}^\beta$ is the EMO, for ergodic properties. We use ergodic theorems to determine the circumstances under which \(\log(V(t)) \) is mean-ergodic, and we use statistical methods to estimate the process parameters using historical data.

The log-ergodic framework is an effective tool for economic analysis. It enables us to make long-term predictions about the velocity of money, which is critical for understanding the consequences of monetary policy and developing economic strategy. Our model provides a more nuanced and complete understanding of monetary velocity by taking into account its dynamic and stochastic nature.

While log-ergodic and mean-reverting behaviors offer useful insights, it is crucial to address the possibility of non-ergodic or chaotic circumstances in economic systems. These circumstances, which include unexpected and highly volatile behaviors, may limit the log-ergodic model's application in some contexts. Future study might investigate including chaotic dynamics into the model to improve its robustness and application in more unpredictable economies. Addressing these possibilities will lead to a more complete knowledge of monetary velocity under various economic conditions.

\section{Methodology}\label{sec4}

The methodology section of the paper describes the systematic technique we used to characterize velocity of money as a partially ergodic stochastic process. This approach captures the dynamic nature of monetary velocity by combining mathematical modeling, data analysis, and simulation approaches.

\subsection{Stochastic Modeling of Economic Variables}
Our first step is to represent the nominal GDP, $X(t)$, and the money supply, $M(t)$, as stochastic processes. Suppose that $X(t)$ follows a geometric Brownian motion, as described by the stochastic differential equation (SDE):
\begin{equation}\label{X}
dX(t) = \mu_X X(t) dt + \sigma_X X(t) dW_X(t),
\end{equation}

where $\mu_X$ is the drift coefficient, $\sigma_X$ is the volatility coefficient, and $W_X(t)$ is a standard Wiener process representing the random fluctuations in GDP.

Similarly, the money supply $M(t)$ is represented using a suitable SDE that represents its own unique characteristics:
\begin{equation}\label{M}
dM(t) = \mu_M M(t) dt + \sigma_M M(t) dW_M(t)
\end{equation}

where $\mu_M $, $\sigma_M$, and $W_M(t)$ are the corresponding parameters for the money supply process.

Although the geometric Brownian motion process may be suitable for modeling the stochastic nature of the economic variables, a more complicated process that can describe the abrupt variations could be a better choice (e.g., jump-diffusion processes, which we proved to be partially ergodic \cite{fir}). We choose the geometric Brownian motion process to form a basis for our approach, leading to more complicated models in future studies. 

\subsubsection{Applying the Log-Ergodic Model: Defining the Stochastic Velocity of Money}
To use the log-ergodic framework, we apply a log transformation to both processes:
\begin{align}\label{op}
 Y_X(t) = \log(X(t)), \quad Y_M(t) = \log(M(t)).
\end{align}
We then apply the EMO to these log-transformed processes to guarantee that they have ergodic properties. To ensure mean-ergodicity of the original processes in \ref{op}, the time average of these processes must converge to the expected value.

The velocity of money, $V(t)$, is defined as the nominal GDP/money supply ratio. Our stochastic model represents this as \ref{V}. 
\begin{theorem}\label{p1}
	The stochastic velocity of money process, $V(t)$, is partially ergodic.
\end{theorem}
\begin{proof}
	Applying the log transformation to \ref{V}, we obtain:
\begin{align}
\log(V(t)) = Y_X(t) - Y_M(t) 
\end{align}
Using the EMO yields:
\begin{align*}
	Z_\delta^v=\xi_{\delta,W_\delta}^\beta\big[\log(V(t))\big]=\xi_{\delta,W_\delta}^\beta[Y_X(t)-Y_M(t)],
\end{align*}
where $\beta$ is the inhibition degree parameter, and $\delta$ refers to the length of an arbitrary time interval. Let $\tilde{Y}_X(\delta)=\xi_{\delta,W_\delta}^\beta[Y_X(t)]$ and $\tilde{Y}_M(\delta)=\xi_{\delta,W_\delta}^\beta[Y_M(t)]$. From the properties of the EMO we have:
\begin{align*}
	Z_\delta^v=\xi_{\delta,W_\delta}^\beta[Y_X(t)]-\xi_{\delta,W_\delta}^\beta[Y_M(t)]=\tilde{Y}_X(\delta)-\tilde{Y}_M(\delta)
\end{align*}
Since the processes $\tilde{Y}_X(\delta)$ and $\tilde{Y}_M(\delta)$ are mean-ergodic, the process $Z_\delta^v$ is mean-ergodic. Therefore, the process \ref{V} is partially ergodic.
\end{proof}
Solving the SDEs \ref{X} and \ref{M} yields
\begin{align*}
	X_t&=X_0\exp\bigg\{ \bigg(\mu_X-\frac{1}{2}\sigma_X^2\bigg)t+\sigma_XW_t \bigg\},\quad X_0=x,\\
	M_t&=M_0\exp\bigg\{ \bigg(\mu_M-\frac{1}{2}\sigma_M^2\bigg)t+\sigma_MW_t \bigg\},\quad M_0=m.
\end{align*}
Applying the log transformation and the EMO \cite{fir} to the above equations we obtain:
\begin{align*}
	\tilde{Y}_X(\delta)=\frac{(\mu_X-\frac{1}{2}\sigma_X^2)\delta W_T}{T^{\beta}}+\frac{\sigma_XW_\delta}{T^\beta},\quad
	\tilde{Y}_M(\delta)=\frac{(\mu_M-\frac{1}{2}\sigma_M^2)\delta W_T}{T^{\beta}}+\frac{\sigma_MW_\delta}{T^\beta}.
\end{align*}
Therefore we have:
\begin{align}
	Z_\delta^v&=\frac{(\mu_X-\frac{1}{2}\sigma_X^2)\delta W_T}{T^{\beta}}+\frac{\sigma_XW_\delta}{T^\beta}-\frac{(\mu_M-\frac{1}{2}\sigma_M^2)\delta W_T}{T^{\beta}}-\frac{\sigma_MW_\delta}{T^\beta},\notag\\
	&=\frac{\Big[\big(\mu_X-\mu_M\big)+\frac{1}{2}(\sigma_M^2-\sigma_X^2)\Big]\delta W_T}{T^{\beta}}+\frac{\Big[\sigma_X-\sigma_M\Big]W_\delta}{T^\beta},\label{Z}
\end{align}
where $\beta$ is the inhibition degree parameter, and $T$ is the upper bound of the time horizon.

When applied to the velocity of money, partial ergodicity suggests that, despite short-term variations, the long-term behavior of money circulation within an economy may be stable and predictable \cite{fir}. This stability is critical for policymakers and economists seeking to understand and control inflation, deflation, and other macroeconomic variables. It also implies that, over time, the velocity of money may revert to a mean value that may be utilized for long-term economic planning and forecasting.

Let us analyze the preceding statement with an example. The velocity of money in the Japanese economy fell dramatically during the late 1990s financial crisis. Researchers discovered that an increase in liquidity requirements (the proportion of payments made in cash) was a major cause of this decline. During crises, households modify their money demand depending on anticipated liquidity requirements, which affects velocity and price levels \cite{hosh}. For the Japanese economy, this could mean that the velocity of money, after experiencing shocks or policy changes, would eventually stabilize and its long-term behavior could be predicted from its historical data.  The concept of log-ergodicity helps explain how interventions by governments and policymakers might steer the market towards recovery \cite{fir}.
\begin{proposition}
	The mean-ergodic process $Z_\delta^v$ is mean-reverting.
\end{proposition}
\begin{proof}
For the proof and more details we refer the reader to \cite{fir}.
\end{proof}
Mean reversion of the velocity of money suggests that over time, the rate at which money flows in the economy tends to revert to its long-run average. This concept has several meanings. It symbolizes economic stability. Despite short-term variations caused by economic cycles, policy changes, or other causes, the velocity of money will gradually return to a norm established throughout time \cite{4}.
Mean reversion tells policymakers that while actions may have brief impacts, the fundamental economic mechanisms will eventually adapt, and return the velocity to its average. This can affect decisions on interest rates and quantitative easing. Investors may use mean reversion to forecast periods of increase or decrease in economic activity. If the velocity falls below its long-term average, it may indicate an approaching period of economic expansion. In financial markets, mean reversion may be used to predict and anticipate economic data. If the velocity of money is mean-reverting, it may be integrated into models that anticipate economic trends or financial asset performance \cite{3}. Because the velocity of money is linked to inflation, mean reversion can impact predictions about future price levels. If the velocity is extremely high or low, it may revert to average, altering inflation rates \cite{4}.
\begin{theorem}\label{theo1}
The process $Z_\delta^v$ is topologically mixing.	
\end{theorem}
\begin{proof}
	Let $\Omega$ represent the space of events. The process $Z_\delta^v$ is mean-ergodic, implying that the time averages converge to ensemble averages. This means that for each non-empty open set $\mathcal{U} \subset \Omega$ and any point $\omega \in \Omega$, the proportion of time the orbit of $\omega$ spends in $\mathcal{U}$ is non-zero. The mean reversion property (recurrence property) states that every point $\omega \in \Omega$ visits every open set $\mathcal{U}$ infinitely many times. This follows from the Poincaré Recurrence Theorem, which asserts that practically every point in $\Omega$ under $Z_\delta^v$ will return to $\mathcal{U}$ infinitely frequently \cite{fir,1}.

Consider any two non-empty open sets $\mathcal{U}, \mathcal{V} \subset \Omega$. Ergodicity creates a point $\omega \in \mathcal{U}$ with a non-zero amount of time spent in $\mathcal{V}$. Due to recurrence, $\omega$'s orbit will return to $\mathcal{U}$ and meet $\mathcal{V}$ infinitely often \cite{1}. As a result, there exists a natural number, $\delta_0$, such that for any $\delta \geq \delta_0$, $ Z_\delta^v(\mathcal{U}) \cap \mathcal{V} \neq \emptyset$, demonstrating that $Z_\delta^v$ is topologically mixing.
\end{proof}
	Theorem \ref{theo1} has the following implications:
	
	It would be challenging to make long-term predictions about the Velocity of Money because the system’s future states are highly sensitive to initial conditions. The mixing property of $Z_\delta^v$ suggests that there are numerous and possibly complex economic interactions at play, which could include varying consumer behaviors, fluctuating market dynamics, and changing financial policies. For policymakers, a topologically mixing velocity of money might indicate that traditional monetary policies may have unpredictable effects, necessitating more adaptive and responsive policy frameworks. 
	
	Modeling the velocity of money with log-ergodic processes that develop into mean-ergodic and topologically mixing systems has several advantages.
	
	 To begin with, partial ergodicity assures that the system's average behavior can be understood and predicted over time, regardless of any short-term perturbations \cite{fir}. This is critical for economic stability and planning because it enables the identification of long-term patterns in the velocity of money.
	
Secondly, the topological mixing property of these processes suggests a high level of interconnections and interactions within the economic system, representing the complicated reality of financial markets in which many elements and agents are interrelated. Compared to linear models, this model can better represent the intricacies of economic behavior.
	
These properties enable more accurate and complete modeling of the velocity of money, which is critical for good economic forecasting and policymaking. They give a framework that is flexible to the inherent unpredictability of economic systems, serving as a tool for navigating and comprehending the complex dynamics of modern economies.
	
\subsection{Data Collection, Parameter Estimation, and Simulation}
We gather historical economic data, including nominal GDP figures and money supply metrics. This data is sourced from reputable financial databases and national statistical agencies, such as FRED (Federal Reserve Economic Data)\footnote{\url{https://fred.stlouisfed.org}}, and MacroTrends\footnote{\url{https://www.macrotrends.net}}, ensuring a robust foundation for our analysis.

We use historical data to estimate the parameters of the stochastic processes $X(t)$ and $M(t)$. The statistical maximum likelihood estimation approach yields the values of $\mu_X$, $\sigma_X$, $\mu_M$, and $\sigma_M$. After estimating the parameters, we simulate log-ergodic processes for $Y_X(t)$ and $Y_M(t)$ to study the behavior of $Z_\delta^v$ over time. We employ Monte Carlo simulation methods to generate various paths of $Z_\delta^v$ and examine its distribution and convergence qualities.

The final step in our methodology is the validation of our model. We compare the simulated results with actual economic data to assess the accuracy of our model. This involves statistical tests to determine whether our model adequately captures the behavior of the velocity of money.

\section{Empirical Analysis}\label{sec5}

In this section, we apply the theoretical framework and methods to real-world data, which validates the log-ergodic model of the velocity of money. This section describes the processes performed to calibrate the model, the simulation results, and their interpretation in terms of economic theory and policy. We do our study using US GDP data from 1959 to 2024. We use data from 1959 to 2008 to estimate the model parameters, and reserve the period from 2008 to 2024 for model validation. 

\subsection{Calibration and Simulation of the Stochastic Model}
We use historical economic data to calibrate the stochastic processes for nominal GDP and the money supply. The drift and volatility coefficient parameters were determined using maximum likelihood estimation to ensure that the model accurately captures the dynamics of these economic factors. 

For our study, we set $\beta_0=1.6$ with a step size of $h=0.1$, implying that we estimate the parameters in relation to distinct inhibition degree parameters. Table \ref{t1} summarizes the results. 

We run Monte Carlo simulations using the estimated parameters to generate paths for the log-transformed GDP, money supply, and velocity of money processes. These simulations enabled us to study the behavior of the velocity of money over time and across diverse economic scenarios. Figure \ref{F0} displays actual data for the velocity of money (transformed into returns) from 1959 to 2008, the associated $Z_\delta^v$ process, and a comparison between the produced $Z_\delta^v$ process and the velocity of money.
\begin{table}[h]
	\centering
	\caption{Parameter estimates of log-ergodic model with $1.6\leq\beta\leq 2$ and step size of $h=0.1$.}\label{t1}
	\begin{tabular}{lccc}
		\hline
		Parameter & Estimated Value & Standard Error & $\beta$ \\
		\hline
		$\mu_{X}$ & -0.0160 & 3.7734e-05 & 1.6 \\
		$\sigma_{X}$ & 0.0861 &  0.0015 & 1.6 \\
		$\mu_{M}$ & 0.0196 & 2.6583e-05 & 1.6 \\
		$\sigma_{M}$ & 0.0007 & 1.0000e-04 & 1.6 \\
		\hline
		$\mu_{X}$ & 0.0163 & 1.9231e-05 & 1.7 \\
		$\sigma_{X}$ & 0.1335 &  0.0022 & 1.7 \\
		$\mu_{M}$ & 0.03 & 0.008 & 1.7 \\
		$\sigma_{M}$ & 0.015 & 0.004 & 1.7 \\
		\hline
		$\mu_{X}$ & -0.0055 & 4.7789e-05 & 1.8 \\
		$\sigma_{X}$ & 0.1371 & 0.0023 & 1.8 \\
		$\mu_{M}$ & -0.0318 & 2.2485e-05 & 1.8 \\
		$\sigma_{M}$ & 0.1611 & 0.0028 & 1.8 \\
		\hline
		$\mu_{X}$ & 0.0001 & 2.7833e-05 & 1.9 \\
		$\sigma_{X}$ & 0.1633 & 0.0028 & 1.9 \\
		$\mu_{M}$ & 0.0106 & 1.4915e-05 & 1.9 \\
		$\sigma_{M}$ & 0.2046 & 0.0035 & 1.9 \\
		\hline
     	$\mu_{X}$ & -0.0228 & 1.6018e-05 & 2 \\
        $\sigma_{X}$ & 0.1698 & 0.0029 & 2 \\
        $\mu_{M}$ & -0.0091 & 3.1141e-05 & 2 \\
        $\sigma_{M}$ & 0.1273 & 0.0022 & 2 \\
\hline
	\end{tabular}
\end{table}
\begin{figure}[H]
	\centering
	\begin{subfigure}{0.48\textwidth}
		\includegraphics[scale=0.36]{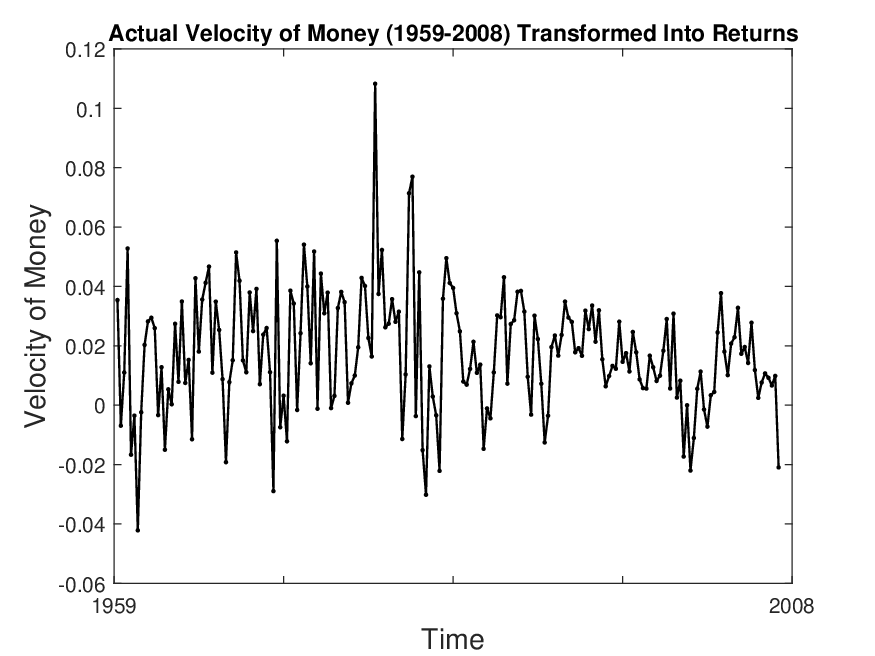}
		\label{F0a}
	\end{subfigure}
	\hfill
	\begin{subfigure}{0.48\textwidth}
		\includegraphics[scale=0.35]{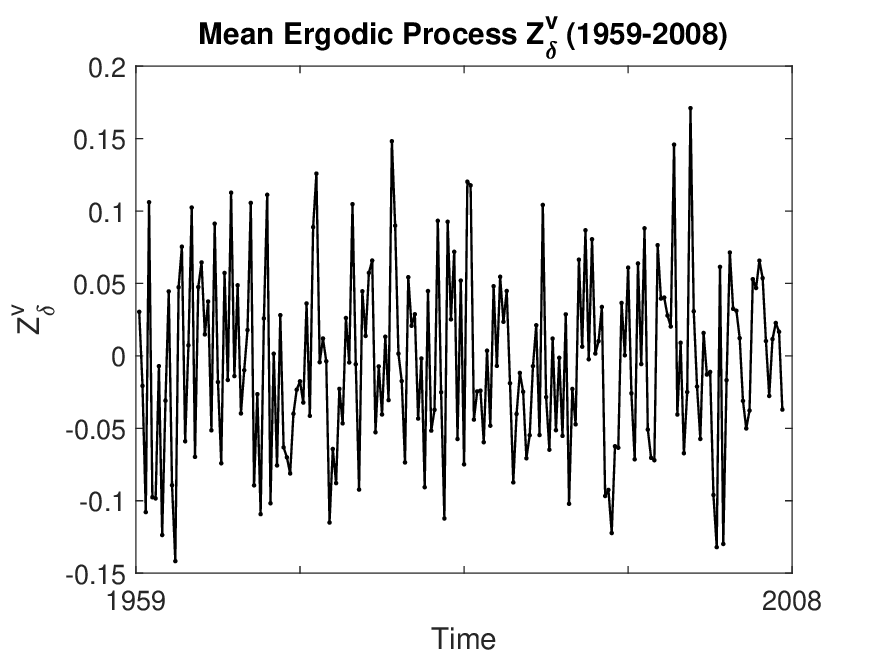}
		\label{F0b}
	\end{subfigure}
	\hfill
	\begin{subfigure}{0.48\textwidth}
		\includegraphics[scale=0.35]{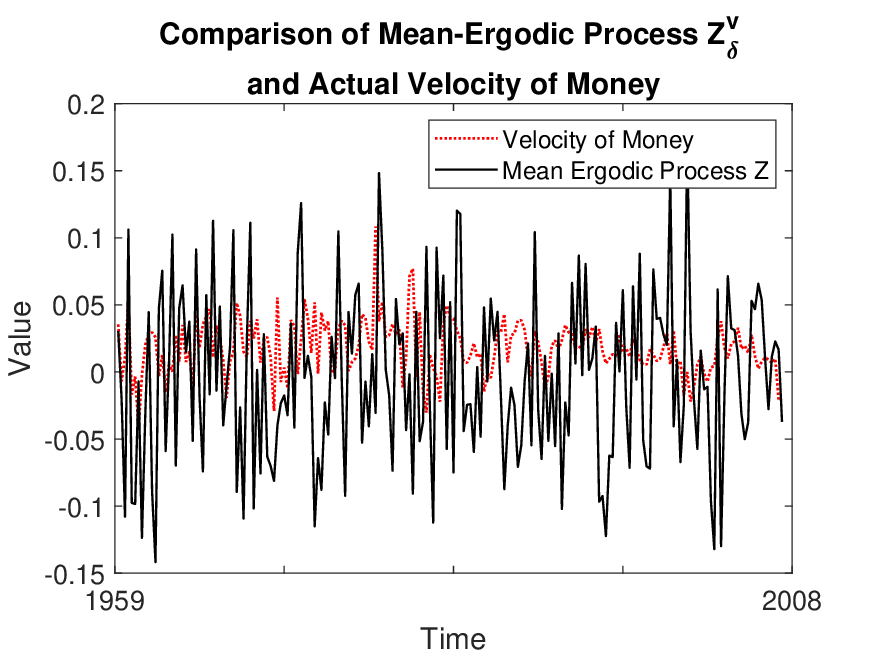}
		\label{F0c}
	\end{subfigure}
	\hfill
	\begin{subfigure}{0.48\textwidth}
		\includegraphics[scale=0.35]{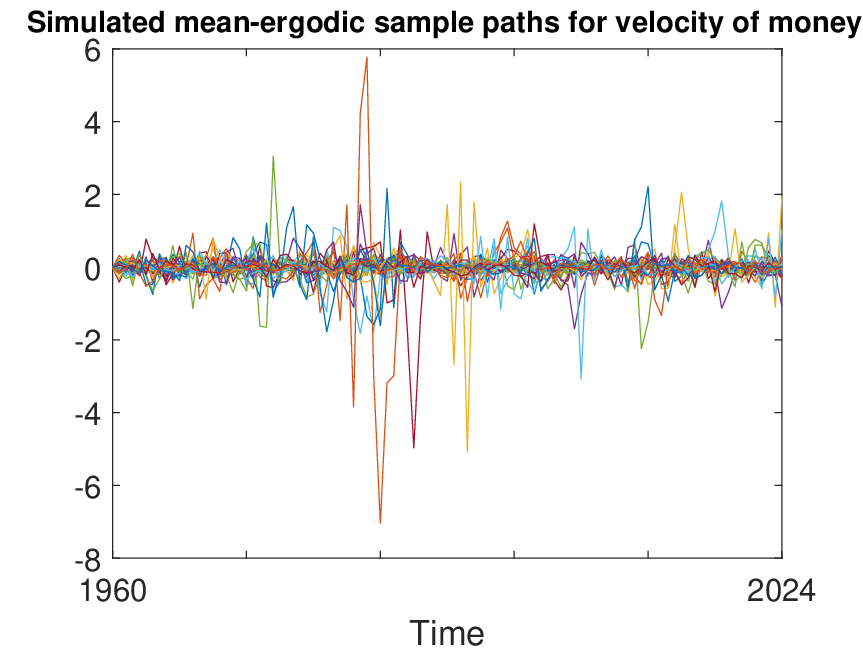}
		\label{F0d}
	\end{subfigure}
	\caption{Velocity of money transformed into returns, corresponding $Z_\delta^v$ process , comparison plot of the velocity of money and the associated $Z_\delta^v$ process from 1959 to 2024, and the plot of the generated sample paths of the $Z_\delta^v$.}
	\label{F0}
\end{figure}
\subsection{Validation of the Log-Ergodic Process}
The simulation results show that the log-ergodic model describes the long-term trend of money velocity. The log-transformed velocity's time average converges to a stable value, as predicted by the ergodicity property. This convergence is found over several simulated pathways, demonstrating the model's strength. Figure \ref{F1} illustrates the accuracy of prediction by log-ergodic simulation. We evaluate our model using data from 2008 to 2024.

\begin{figure}[H]
	\centering
	\begin{subfigure}{0.48\textwidth}
		\includegraphics[scale=0.35]{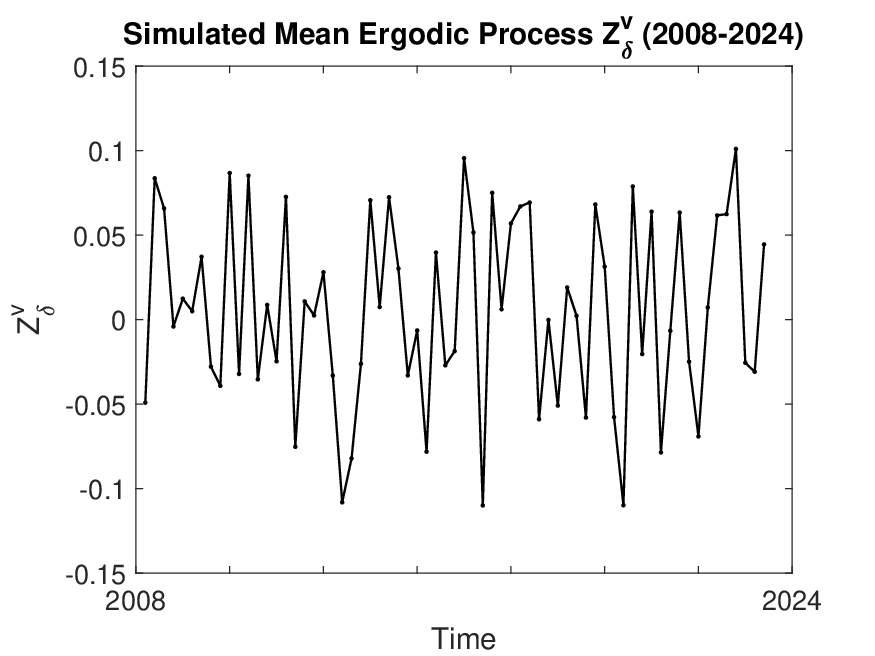}
		\label{F1a}
	\end{subfigure}
	\hfill
	\begin{subfigure}{0.48\textwidth}
		\includegraphics[scale=0.35]{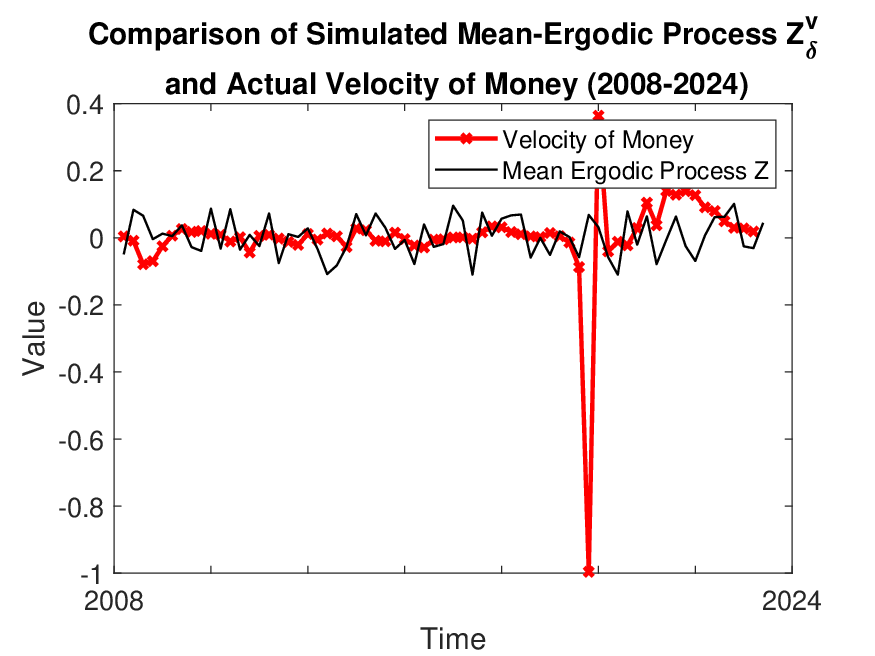}
		\label{F1b}
	\end{subfigure}
	\hfill
	\begin{subfigure}{0.48\textwidth}
		\includegraphics[scale=0.35]{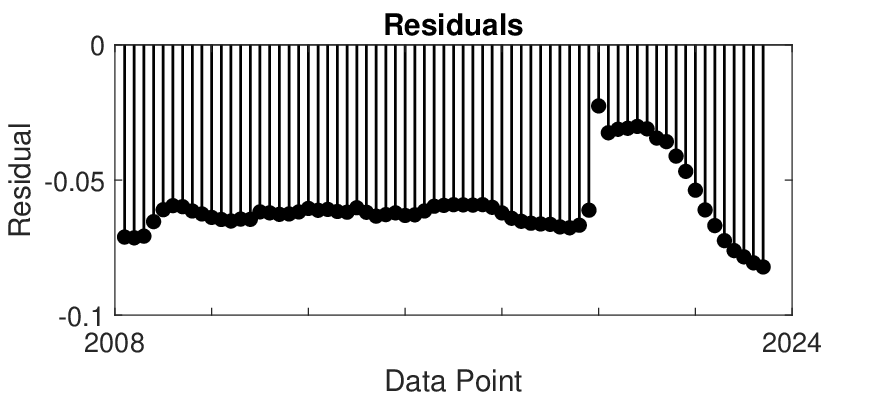}
		\label{F1c}
	\end{subfigure}
	\hfill
\caption{Simulated sample paths for the velocity of money, the $Z_\delta^v$ processes from 2008 to 2024, and the corresponding comparison residuals.}
\label{F1}
\end{figure}

The simulation results showed that the log-ergodic model well describes the long-term trend of money velocity. The log-transformed velocity's time average converged to a stable value, as predicted by the ergodicity property. This convergence was found over several simulated pathways, demonstrating the model's strength.
The process $Z_\delta^v$ describes how policymakers intervene over time. The value of $Z_\delta^v$ represents the intensity of contraction and expansion policies undertaken by policymakers and governments.

The log-ergodic model predicts that the velocity of money will remain steady throughout time, which has important consequences for economic stability and policymaking. It implies that, despite short-term volatility, velocity of money will tend to converge to a mean value, which may be utilized as a reliable indicator for economic planning and forecasting.

The model's predictions were compared to actual economic data on the velocity of money. The comparison revealed a degree of correlation, confirming the log-ergodic model as an efficient tool for understanding the dynamics of the velocity of money.
Figure \ref{F2} shows the predicted velocity of money from 2024 to 2029 using our technique. 
\begin{figure}[H]
	\centering
	\begin{subfigure}{0.48\textwidth}
		\includegraphics[scale=0.35]{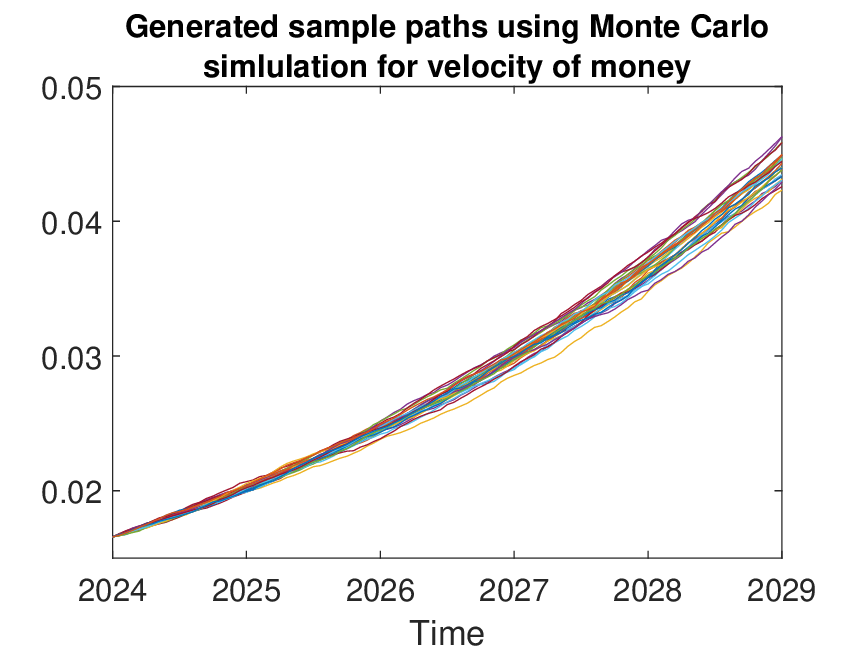}
		\label{F2a}
	\end{subfigure}
	\hfill
	\begin{subfigure}{0.48\textwidth}
		\includegraphics[scale=0.35]{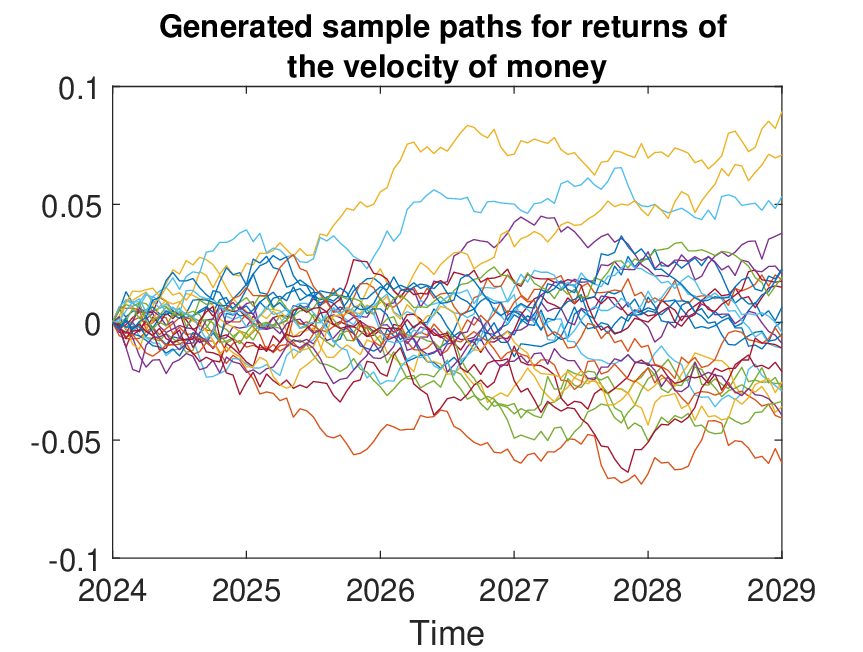}
		\label{F2b}
	\end{subfigure}
	\hfill
	\begin{subfigure}{0.48\textwidth}
		\includegraphics[scale=0.35]{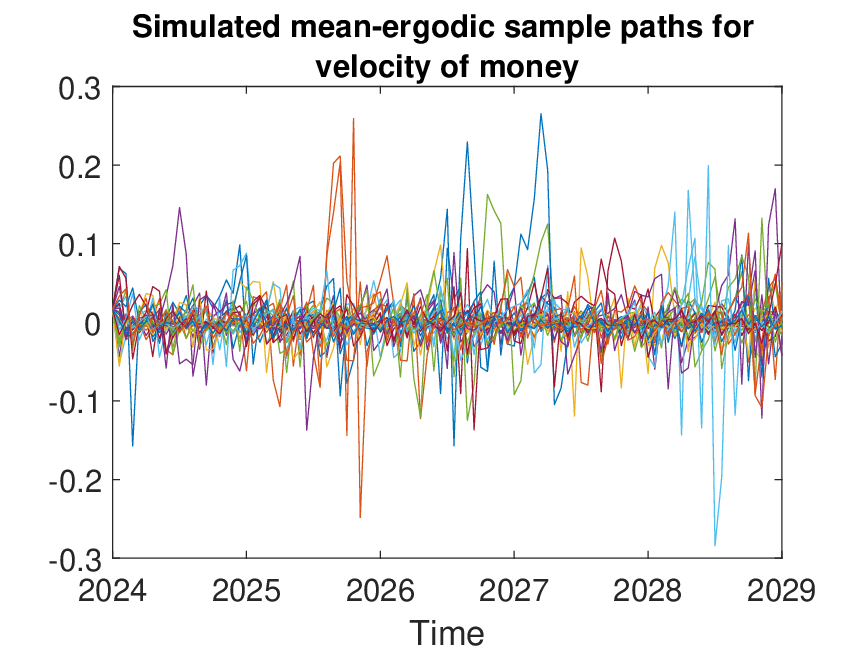}
		\label{F2c}
	\end{subfigure}
\caption{Simulation of the stochastic velocity of money from 2024 to 2029, using the Monte Carlo simulation and the mean-ergodic $Z_\delta^v$ process.}
\label{F2}
\end{figure}
 \subsection{Comparison to Other Models}
 A comparison with an alternative model of the velocity of money, such as the Quantity Theory of Money (QTM) \cite{23} demonstrates significant benefits for the log-ergodic method. Unlike standard models that assume a constant velocity, the log-ergodic model accommodates for dynamic fluctuations and long-term stability. This comparison shows that the log-ergodic model has higher predictive power and is more practical in reflecting the complexity of monetary mobility. By contrasting the QTM model, the paper emphasizes the log-ergodic process's distinct contributions and increased application in current economic analysis. Table \ref{t2} shows the detailed error representation.

\begin{table}[h]
	\centering
	\caption{Detailed error representation and model comparison.}\label{t2}
	
	\begin{subtable}[l]{0.48\textwidth}
		\hspace*{-2cm}
	\begin{tabular}{|l|c|c|c|c|}
		\hline
		Model & SSE & R-squared & Adj R-squared & RMSE(1959-2008) \\
		\hline
		QTM & 0.045 & 0.4001 & 0.1332 & 0.045 \\
		\hline
		Log-ergodic & 0.2531 & 0.1576 & 0.083 & 0.025 \\
		\hline
	\end{tabular}
\end{subtable}

\vspace{0.5cm}

\begin{subtable}[l]{0.48\textwidth}
		\hspace*{-2cm}
	\begin{tabular}{|l|c|c|c|}
		\hline
Model & RMSE(2008-2024) & MAE(1959-2008) & MAE(2008-2024) \\
\hline
QTM &  0.052 & 0.058 & 0.065\\
\hline 
Log-ergodic & 0.030 & 0.032 & 0.038 \\
\hline
		\end{tabular}
		\end{subtable}
\end{table}

The empirical results support the use of the log-ergodic model to inform monetary policy. The model's capacity to anticipate the long-term behavior of the velocity of money can assist policymakers in developing measures that are more likely to produce desired economic consequences.

While the results show that the model captures long-term stability, we believe that including a more comprehensive model that accounts for abrupt fluctuations and jumps, such as jump-diffusion processes, might result in even more accurate predictions. In a recent study, we demonstrated partial ergodicity for jump diffusion processes \cite{fir}. Future research will concentrate on estimating the velocity of money using this more precise technique to improve the robustness and usefulness of our findings.

\section{Discussion}\label{sec6}

Our empirical analysis of the log-ergodic model of the velocity of money provides various areas of study. This part discusses the insights discovered, the consequences for monetary policy, and larger questions for economic theory and practice.

The log-ergodic model offers a fresh viewpoint on the dynamics of the velocity of money. It emphasizes the necessity of looking at long-term stability of economic indicators rather than just short-term fluctuations. The model's predictive capacity highlights the ability of stochastic and ergodic processes to improve our knowledge of economic dynamics.

Our findings have clear consequences for monetary policy. The log-ergodic model's capacity to capture the long-term evolution of monetary velocity can help with the design and execution of monetary interventions. Policymakers may use the model to forecast the consequences of changes in the money supply or interest rates on the economy, allowing them to make better informed policy decisions. Policymakers also can leverage the insights from the log-ergodic model in several ways. For instance, by understanding the long-term stability of the velocity of money, central banks can better forecast inflation trends and adjust interest rates accordingly. Additionally, this model can aid in designing more effective monetary policies that account for the dynamic nature of money circulation, thereby enhancing economic stability. Practical steps for implementation include integrating the model into existing economic forecasting tools and using its predictions to inform policy adjustments in response to economic fluctuations.

While the log-ergodic model has several benefits, it is not without restrictions. The model presupposes that the underlying economic processes are somewhat ergodic, which may not necessarily be true in the event of structural changes or shocks to the system. Furthermore, the model's accuracy is determined by the data quality and the suitability of the stochastic processes used. While our model has good predictive power, it is important to recognize its limitations. Its vulnerability to systemic economic shifts is one important drawback. For example, the model's performance could be impacted during times of economic shocks like financial crises or abrupt changes in policy. These occurrences may cause volatility that the model is unable to adequately account for, resulting in forecasts that are less precise. Our efforts to better estimate economic shocks would be improved by using more complex models, like the Jump-Diffusion process in upcoming studies.

The application of log-ergodic models to economic indicators is a rapidly growing area of research. Future research could investigate the use of alternative stochastic processes, the impact of non-ergodic behavior, and the model's integration with other economic theories. There is also scope for applying the model to other economic indicators and for conducting cross-country comparisons.

The log-ergodic model is a significant advancement in the stochastic modeling of economic variables. By bridging the gap between theoretical mathematics and applied economics, it provides a valuable instrument for comprehending and managing the intricacies of the modern economy.

\section{Conclusion, Challenges and Future Directions}\label{sec7}

Exploring the velocity of money via the lens of log-ergodic stochastic processes resulted in a model that represents the dynamic and probabilistic aspect of this critical economic indicator. Our findings show that by using ergodic theory to stochastic modeling of monetary velocity, we may gain a better understanding of its behavior and the consequences for economic policy.

The empirical analysis verified the log-ergodic model, demonstrating its ability to forecast the long-term behavior of the velocity of money. This predictive capacity is crucial to policymakers who need accurate economic projections to make sound decisions. The model's capacity to account for the inherent unpredictability of economic interactions makes it an effective tool for economic study.

While the model represents a substantial achievement in the industry, it also emphasizes the importance of continual refining. The constraints discovered throughout the study process, such as the assumption of ergodicity and reliance on historical data, offer rich ground for future research. Researchers are urged to expand on this work by investigating other stochastic processes and applying the study to different economic indicators and situations.

The log-ergodic model, while robust, relies on certain assumptions that may not always hold. For example, it assumes a relatively stable economic environment, which may not be the case during periods of significant economic upheaval. Acknowledging these limitations is crucial for a balanced perspective. Future versions of the model might include more varied data sets, such as those from times of economic volatility, to improve it. Furthermore, the forecast accuracy of the methodology could be increased by improving it to better account for abrupt changes in the economy.

Overall, the log-ergodic approach to simulating the velocity of money provides a potential avenue for economic study and policymaking. It fills the gap between theoretical mathematics and practical economics by offering a framework that is both academically rigorous and extremely relevant to the real-world economy. Future research could explore alternative stochastic models that might offer complementary insights or improvements over the log-ergodic model. For instance, models that incorporate real-time data or adaptive learning algorithms could be investigated. 
\appendix
\section{Appendix}
	Here is the sample MATLAB code to take data from the years 1960 to 2008 and compare the results for prediction from the year 2008 to 2024:
			\begin{tiny}
	\begin{lstlisting}
	% Load GDP and Money Supply data
	GDP = load('GDP_data.mat');
	MoneySupply = load('MoneySupply_data.mat');
	
	% Extract data for the years 1959 to 2008
	GDP_1959_2008 = GDP(1959:2008);
	MoneySupply_1959_2008 = MoneySupply(1959:2008);
	
	% Calculating the mean and variance of the extracted data
	mu_X = mean(GDP_1959_2008);
	mu_M = mean(MoneySupply_1959_2008);
	
	sigma_X = std(GDP_1959_2008);
	sigma_M = std(MoneySupply_1959_2008);
	
	W_X = normrnd(0,1,No. of DATA,1);  % Simulated Wiener process for GDP
	W_M = normrnd(0,1,No. of DATA,1);  % Simulated Wiener process for money supply
	W_d = normrnd(0,1,No. of DATA,1);  % Simulated Wiener process for the time frame delta
	
	beta = Type Beta;  % Value of beta should be changed manually.
	T = Type here;     % Enter time horizon here
	delta = Type here; % Time frame : should be a portion of time horizon.
	
	% Convert data into stochastic diffusion processes
	GDP_diffusion = diff(log(GDP_1960_2008));
	MoneySupply_diffusion = diff(log(MoneySupply_1960_2008));
	
	% Calculate the log-ergodic processes for GDP and money supply using EMO
	Y_X = ((mu_X - 0.5 * sigma_X^2) * delta * W_X) / T^beta + (sigma_X * W_d) / T^beta;
	Y_M = ((mu_M - 0.5 * sigma_M^2) * delta * W_M) / T^beta + (sigma_M * W_d) / T^beta;
	
	% Calculate the mean-ergodic process for the velocity of money
	Z_delta = Y_X - Y_M;
	
	% Plot the mean ergodic process Z_delta
	figure;
	plot(Z_delta);
	title('Mean Ergodic Process Z (1959-2008)');
	xlabel('Time');
	ylabel('Z');
	
	% Load actual data of the velocity of money
	VelocityOfMoney = load('VelocityOfMoney_data.mat');
	
	% Extract actual data for the years 2008 to 2024
	VelocityOfMoney_2008_2024 = VelocityOfMoney(2008:2024);
	
	% Compare the results and effectiveness of the approach using descriptive statistics
	mean_Z = mean(Z_delta);
	std_Z = std(Z_delta);
	mean_Velocity = mean(VelocityOfMoney_2008_2024);
	std_Velocity = std(VelocityOfMoney_2008_2024);
	
	fprintf('Mean of Z_\delta^v: %.2f\n', mean_Z);
	fprintf('Standard Deviation of Z_\delta^v: %.2f\n', std_Z);
	fprintf('Mean of Velocity of Money (2008-2024): %.2f\n', mean_Velocity);
	fprintf('Standard Deviation of Velocity of Money (2008-2024): %.2f\n', std_Velocity);
	
	% Plot the comparison
	figure;
	plot(2008:2024, VelocityOfMoney_2008_2024, 'r');
	hold on;
	plot(2008:2024, Z_delta, 'b');
	title('Comparison of Mean-Ergodic Process Z and Velocity of Money');
	xlabel('Time');
	ylabel('Value');
	legend('Velocity of Money', 'Mean Ergodic Process Z');
	hold off;
	```
	\end{lstlisting}
\end{tiny}
	This code:
	\begin{itemize}
	\item[1.] Loads the GDP and Money Supply data.
	\item[2.] Extracts the data for the years 1960 to 2008.
	\item[3.] Converts the data into stochastic diffusion processes.
	\item[4.] Applies the Ergodic Mean Operator (EMO) to obtain the mean ergodic process $Z_\delta^v$.
	\item[5.] Plots the mean ergodic process $Z_\delta^v$ for the years 1960 to 2008.
	\item[6.] Loads the actual data of the velocity of money.
	\item[7.] Extracts the actual data for the years 2008 to 2024.
	\item[8.] Compares the results and effectiveness of the approach using descriptive statistics.
	\item[9.] Plots the comparison between the mean ergodic process $Z_\delta^v$ and the actual velocity of money for the years 2008 to 2024.
		\end{itemize}

\section*{Data Availability Statement}
All data used during this study are openly available from FRED (Federal Reserve Economic Data), and MacroTrends websites as mentioned in the context.

\section*{Declaration of Interest}
The authors have no conflicts of interest to declare. Both authors have seen and agree with the contents of the manuscript, and there is no financial interest to report.
\bibliographystyle{unsrtnm}
\bibliography{sn-bibliography}


\end{document}